\documentclass[11pt]{article}
\usepackage[utf8]{inputenc}
\usepackage{amsfonts}
\usepackage{amssymb}
\usepackage{fullpage,tikz,enumitem}

\usepackage[colorlinks=true, pdfstartview=FitV, 
            linkcolor=black, citecolor=black, 
            urlcolor=black, plainpages=false,
            pdfpagelabels]{hyperref}

\usepackage{bookmark}
\usepackage{amsmath,amsthm,amssymb,color,verbatim,graphicx,fullpage,url}
\usepackage{mathtools}
\usepackage{thm-restate}
\usepackage[capitalise]{cleveref}

\usepackage{amsmath,amsthm, amssymb,amscd,amstext,amsfonts}
\usepackage[utf8]{inputenc}
\usepackage{mathtools}
\usepackage{mathrsfs}
\usepackage{thmtools}
\usepackage{latexsym}
\usepackage{verbatim}
\usepackage{framed}
\usepackage{graphicx}
\usepackage{stmaryrd}
\usepackage{enumitem}
\usepackage{fullpage}
\usepackage{bm}
\usepackage{bbm}
\usepackage{url}
\usepackage{physics}
\usepackage{dsfont}
\usepackage{todonotes}
\usepackage{tikz}
\usepackage{forest}

\usepackage{color}
 
\theoremstyle{plain}

\newtheorem{thm}{Theorem}[section]
\newtheorem{theorem}[thm]{Theorem}

\newtheorem{conjecture}[thm]{Conjecture}

\newtheorem{lemma}[thm]{Lemma}

\newtheorem{example}[thm]{Example}
\newtheorem{fct}[thm]{Fact}

\newtheorem{claim}[thm]{Claim}
\newtheorem*{claim*}{Claim}

\theoremstyle{remark}
\newtheorem{remark}[thm]{Remark}

\def\sign{\mathrm{sign}} 
\def\S{\mathcal{S}} 
\def\F{\mathbb{F}} 
\def\R{\mathbb{R}} 
\def\N{\mathcal{N}} 
 
\def\1{\mathbf{1}} 
\def\0{\mathbf{0}}

\DeclareMathOperator{\CC}{\mathrm{CC}}
\DeclareMathOperator{\PDT}{\mathrm{PDT}}

\DeclareMathOperator{\poly}{poly}
\DeclareMathOperator{\polylog}{polylog}

 \newcommand{\arXiv}{ar\kern-0.5ex\raisebox{0.5ex}{X}iv}

\newcommand{\supp}{\mathrm{supp}}

\title{Refuting approaches to the log-rank conjecture for XOR functions}

\author{
Hamed Hatami\thanks{School of Computer Science, McGill University. Email: \texttt{hatami@cs.mcgill.ca}}
\and
Kaave Hosseini\thanks{Department of Computer Science, University of Rochester. Email: \texttt{kaave.hosseini@rochester.edu}}
\and
Shachar Lovett\thanks{Department of Computer Science and Engineering, UC San Diego. Research supported by NSF awards 1953928 and 2006443, and a Simons investigator award. Email: \texttt{slovett@ucsd.edu}}
\and
Anthony Ostuni\thanks{Department of Computer Science and Engineering, UC San Diego. Research supported by NSF award 2006443. Email: \texttt{aostuni@ucsd.edu}}
}
\date{}

\begin{document}

\maketitle

\begin{abstract}
    The log-rank conjecture, a longstanding problem in communication complexity, has persistently eluded resolution for decades.   Consequently, some recent efforts have focused on potential approaches for establishing the conjecture in the special case of XOR functions, where the communication matrix is lifted from a boolean function, and the rank of the matrix equals the Fourier sparsity of the function, which is the number of its nonzero Fourier coefficients.	
	
    In this note, we refute two conjectures. The first has origins in Montanaro and Osborne (arXiv'09) and is considered in Tsang, Wong, Xie, and Zhang (FOCS'13), and the second is due to Mande and Sanyal (FSTTCS'20). These conjectures were proposed in order to improve the best-known bound of Lovett (STOC'14) regarding the log-rank conjecture in the special case of XOR functions. Both conjectures speculate that the set of nonzero Fourier coefficients of the boolean function has some strong additive structure. We refute these conjectures by constructing two specific boolean functions tailored to each.
\end{abstract}

\section{Introduction}
The study of communication complexity seeks to determine the inherent amount of communication between multiple parties required to complete a computational task. Arguably, the most outstanding conjecture in the field is the \textit{log-rank conjecture} of Lov\'{a}sz and Saks \cite{lovasz1993communication}. They suggest that the (deterministic) communication complexity of a two-party boolean function is upper bounded by the matrix rank over $\R$. More precisely, 
\begin{conjecture}[Log-rank conjecture \cite{lovasz1993communication}]\label{conj:LRC}
    Let $f : X \times Y \to \{-1,1\}$ be an arbitrary two-party boolean function. Then,
    \[
    \CC(f) \le \polylog(\rank(f)),
    \]
    where $\CC(f)$ is the communication complexity of $f$ and $\rank(f)$ is the rank over $\R$ of the corresponding boolean matrix.
\end{conjecture}
It is well-known that $\log(\rank(f)) \le \CC(f)$ \cite{mehlhorn1982vegas}, so a positive resolution to \cref{conj:LRC} would imply that the communication complexity of two-party boolean functions is determined by rank, up to polynomial factors.

To date, the best known bound is still exponentially far from that in \cref{conj:LRC}. Concretely, Lovett \cite{lovett2016communication} showed that $\CC(f) \le O(\sqrt{\rank(f)} \log \rank(f))$. Very recently, Sudakov and Tomon posted a preprint improving the bound to $O(\sqrt{\rank(f)})$ \cite{sudakov2023matrix}. 
In hopes of gaining further insight, many researchers have considered the special case of \textit{XOR functions}, where $f_{\oplus}(x,y) \coloneqq f(x + y)$ for a boolean function $f : \mathbb{F}_2^n \to \{-1,1\}$ \cite{montanaro2009communication, zhang2010parity, tsang2013fourier, shpilka2014structure, hatami2018structure, mande2020parity}. 

The XOR setting has several convenient properties. For example, the eigenvalues of $f_{\oplus}$ correspond to the Fourier coefficients of $f$. Thus, $\rank(f_{\oplus}) = |\supp(\widehat{f})|$, the number of nonzero coefficients in $f$'s Fourier expansion (also known as the \textit{Fourier sparsity}). Additionally, Hatami, Hosseini, and Lovett \cite{hatami2018structure} proved a polynomial equivalence between $\CC(f_{\oplus})$ and the \textit{parity decision tree} complexity of $f$, denoted $\PDT(f)$. Parity decision trees are defined similarly to standard decision trees, with the extra power that each node can query an arbitrary parity of input bits. These facts together imply that the log-rank conjecture for XOR functions can be restated as follows:
\begin{conjecture}[XOR log-rank conjecture]\label{conj:XORLRC}
    Let $f : \F_2^n \to \{-1,1\}$. Then,
    \[
    \PDT(f) \le \polylog(|\supp(\widehat{f})|).
    \]
\end{conjecture}
The best known bound, due to  \cite{tsang2013fourier, shpilka2014structure}, is $\PDT(f) \le O(\sqrt{|\supp(\widehat{f})}|)$, a mere $\log$-factor improvement on the general case bound by Lovett \cite{lovett2016communication}, and matched by the recent bound of Sudakov and Tomon \cite{sudakov2023matrix}.

\subsection{Folding}
\textit{Folding} is a fundamental concept in the analysis of the additive structure of a function's Fourier support. Let 
\[
\S = \supp(\widehat{f}) = \{\gamma \in \F_2^n: \widehat{f}(\gamma)\neq 0\} \quad \text{and} \quad \S + \gamma = \{s + \gamma : s \in \S\}.
\]
If $(s_1, s_2), (s_3, s_4) \in \binom{\S}{2}$ satisfy $s_1 + s_2 = s_3 + s_4 = \gamma$, we say the pairs $(s_1, s_2)$ and $(s_3, s_4)$ \textit{fold} in the direction $\gamma$. 

Analyzing folding directions is useful in constructing efficient PDTs in the context of \Cref{conj:XORLRC}. In particular, when a function $f$ is restricted according to the result of some parity query $\gamma$, all pairs of elements in $\S$ that fold in the direction $\gamma$ collapse to a single term in the restricted function $f|_{\gamma}$'s Fourier support. Iterating this process until the restricted function is constant yields a PDT whose depth depends on the number of iterations performed and, thus, on the size of the folding directions queried. Indeed, this is the general strategy used to prove the aforementioned closest result to \Cref{conj:XORLRC} \cite{tsang2013fourier, shpilka2014structure}.

\subsection{Refuting a greedy approach}
An approach dating back to \cite{montanaro2009communication} seeks to prove \cref{conj:XORLRC} through the existence of a single large folding direction. 
They conjectured that there always exists $\gamma_1, \gamma_2$ such that $|(\S + \gamma_1) \cap (\S + \gamma_2)| \ge |\mathcal{S}|/K$ for some constant $K > 1$. This yields the following $O(\log |\S|)$-rounds greedy approach: query $\gamma_1 + \gamma_2$ and consider the function restricted to the query response. This restriction decreases the Fourier sparsity by a constant factor, so the function must become constant in $O(\log |\S|)$ rounds. This implies the strong upper bound of 
\[
    \PDT(f) \le O\left(\log|\mathcal{S}|\right).
\]

However, O'Donnell, Wright, Zhao, Sun, and Tan \cite{odonnell2014composition} constructed a function with communication complexity $\Omega(\log(|\S|)^{\log_3(6)})$; hence one can not take $K$ to be a constant. Yet to prove the log-rank conjecture, it suffices to take $K = O(\polylog(|\S|))$, and this choice of $K$ remained plausible up to date.  Such an approach is mentioned in both \cite{tsang2013fourier} and \cite{mande2020parity}, and a similar approach was used to verify the log-rank conjecture for many cases of functions lifted with AND (rather than XOR) gadgets \cite{knop2021log}. We strongly refute this conjecture.
\begin{theorem}[Informal version of \Cref{thm:preciseGreedy}]\label{thm:nogreedy}
    For infinitely many $n$, there is a function $f:\F_2^n\to\{-1,1\}$ such that for $\S = \supp(\widehat{f})$, it holds
    \[
    |(\S + \gamma_1) \cap (\S + \gamma_2)| \le O\left(|\S|^{5/6}\right)
    \]
    for all distinct $\gamma_1, \gamma_2 \in \F_2^n$.
\end{theorem}
\begin{remark}
    Observe that this theorem implies the greedy method cannot obtain a bound better than $\PDT(f)=\widetilde{O}(|\mathcal{S}|^{1/6})$. In fact, a more careful analysis can rule out bounds better than $\PDT(f) = \widetilde{O}(|\mathcal{S}|^{1/5})$ (see \Cref{rmk:improvedBounds}).
\end{remark}

The functions used in \Cref{thm:nogreedy} are a variant of the addressing function using disjoint (affine) subspaces. While we believe the specific construction is novel, the concept of using functions defined with disjoint subspaces has previously appeared in the literature in this context. Most notably, Chattopadhyay, Garg, and Sherif used XOR functions based on this idea in the pursuit of stronger counterexamples to a more general version of the log-rank conjecture \cite{chattopadhyay2020towards}.

\subsection{Refuting a randomized approach}
Rather than simply looking for a large folding direction, a recent work of Mande and Sanyal \cite{mande2020parity} attempts to address \cref{conj:XORLRC} through a deeper understanding of the additive structure of the spectrum of boolean functions. 
They proposed the following conjecture on the number of nontrivial folding directions, and showed it would yield a polynomial improvement to the state-of-the-art upper bound for the XOR log-rank conjecture via a randomized approach.

\begin{conjecture}[\cite{mande2020parity}]\label{conj:manydir}

There are constants $\alpha, \beta \in (0,1)$ such that for every boolean function $f:\mathbb{F}_2^n\to \{-1,1\}$, for $\S = \supp(\widehat{f})$, it holds
\[
\Pr_{\gamma_1,\gamma_2\in \S}\left[|(\S+\gamma_1)\cap(\S+\gamma_2)|>|\S|^\beta\right]\geq \alpha.
\]
\end{conjecture}
In fact, Mande and Sanyal conjectured that one can take $\beta = \frac{1}{2}-o(1)$. The conjecture might seem plausible given the numerous results on the additive structure of the spectrum of boolean functions. However, we strongly refute it, as well:
\begin{theorem}[Informal version of \Cref{thm:quantifiedMain}]\label{thm:mainthm}
    For infinitely many $n$, there is a boolean function $f:\F_2^n\to\{-1,1\}$ such that for $\S = \supp(\widehat{f})$, it holds
    \[
    \Pr_{\gamma_1,\gamma_2\in \S}\left[|(\S+\gamma_1)\cap(\S+\gamma_2)|> k\right]= O(1/k)\qquad \forall k \ge 1.
    \]
\end{theorem}

\paragraph{Overview.} Some preliminary material is reviewed in \Cref{sec:prelim}. We prove more precise versions of \cref{thm:nogreedy} in \cref{sec:one} and \cref{thm:mainthm} in \cref{sec:many}. \Cref{sec:conc} contains some final thoughts.

\section{Preliminaries}\label{sec:prelim}

\paragraph*{Communication complexity.} 
Let $f : X \times Y \to \{-1,1\}$ be an arbitrary function. Additionally, assume two parties are given an element $x\in X$ and $y\in Y$, respectively, which the other party cannot see. The (deterministic) \emph{communication complexity of $f$}, denoted $\CC(f)$, is the minimum number of bits over all assignments $(x,y)$ needed to be exchanged in order to evaluate $f$, where the parties may decide on a strategy prior to receiving their inputs.

One can view such a function as an $X\times Y$ matrix, where the $(x,y)$ entry takes the value $f(x,y)$. Thus, it is natural to consider the relationship between linear algebraic measures, such as matrix rank, and communication complexity, as in \Cref{conj:LRC}. For a more thorough treatment of communication complexity, see the excellent book \cite{rao2020communication}.

\paragraph*{Decision trees.} 
Decision trees are simple models of computation. The (deterministic) decision tree \textit{depth} of a function $f : \F_2^n \to \{-1,1\}$ is the maximum over all inputs $x\in \F_2^n$ of the fewest number of input bits one must query to correctly evaluate $f(x)$.

Parity decision trees (PDTs) extend the power of ``traditional'' decision trees by allowing queries to return the sum modulo two of an arbitrary subset of the bits. They are particularly relevant in the study of communication complexity, since for functions of the form $f_{\oplus}(x,y) = f(x+y)$ for $f:\F_2^n \to \{-1,1\}$, the parity decision tree depth and communication complexity are equivalent (up to polynomial factors) \cite{hatami2018structure}.

\paragraph*{Boolean analysis.}
Every function $f : \F_2^n \to \mathbb{R}$ has a unique \textit{Fourier expansion}
\[
    f = \sum_{\alpha \in \F_2^n} \widehat{f}(\alpha) \chi_{\alpha},
\]
where
\[
    \chi_{\alpha}(x) = (-1)^{\langle x, \alpha\rangle} \quad\text{and}\quad \widehat{f}(\alpha) = \langle f, \chi_{\alpha} \rangle = \mathbb{E}_{x\in\F_2^n}[f(x)\chi_{\alpha}(x)].
\]
The set $\supp(\widehat{f}) = \{\alpha \in \F_2^n : \widehat{f}(\alpha) \ne 0\}$ is the \textit{Fourier support}, occasionally denoted $\S$. Its size $|\supp(\widehat{f})|$ is the \textit{Fourier sparsity}. In light of \Cref{conj:XORLRC}, we are primarily interested in the relationship between a function's Fourier sparsity and parity decision tree depth.

In general, a vast array of information about a function can be learned from its Fourier expansion, and we direct readers to the standard text \cite{o2014analysis} for additional background. For our purposes, we will only require the following simple fact. Let $V^{\perp} = \{w : \langle w, v \rangle = 0 \text{ for all } v\in V\}$ be the orthogonal complement of a subspace $V$.

\begin{fct}[See e.g., {\cite[Proposition 3.12]{o2014analysis}}]\label{fct:Fourier_subspace}
    If $A = V+v \subseteq \F_2^n$ is an affine subspace of codimension $k$, then
    \[
        \mathbbm{1}_A = \sum_{\alpha \in V^{\perp}} 2^{-k}\chi_{\alpha}(v)\chi_{\alpha}.
    \]
\end{fct}

\section{One excellent folding direction}\label{sec:one}
A large folding direction implies the existence of a parity query whose answer substantially simplifies the resulting restricted function. This suggests the following greedy approach to resolve the XOR log-rank conjecture: if we can always find distinct $\gamma_1,\gamma_2$ such that $|(\S+\gamma_1) \cap (S+\gamma_2)| \ge \Omega(|\S|\:/\polylog(|\S|))$, then querying $\gamma_1+\gamma_2$ and recursing on the appropriate restriction of $f$ will force $f$ to be constant in $\polylog(|\S|)$ rounds.

We refute this strategy by proving a precise version of \cref{thm:nogreedy}.
\begin{theorem}\label{thm:preciseGreedy}
    For $n=2^k+7k$ with $k \in \mathbb{N}^{\ge 3}$, there is a function $f:\F_2^n\to\{-1,1\}$ such that for $\S = \supp(\widehat{f})$, it holds
    $|\S| \ge 2^{6k}$, and yet
    $|(\S + \gamma_1) \cap (\S + \gamma_2)| \le 2^{5k+4}$ for all  distinct $\gamma_1, \gamma_2 \in \F_2^n$.
\end{theorem}

To build intuition for our construction, we first consider the standard addressing function.
\begin{example}[Addressing]
    Define $f : \F_2^{k+2^k} \to \{-1,1\}$ by 
    \[
    f(x,y) = (-1)^{y_x} = \sum_{z\in \F_2^k} \mathbbm{1}_z(x) \cdot (-1)^{y_z},
    \]
    where $x \in \F_2^{k}$ and $y \in \F_2^{2^k}$ (and slightly abusing notation by indexing $y$ with vectors).
\end{example}

A greedy approach is sufficient for a PDT to evaluate this function. Simply query each address bit, then the corresponding addressed bit. Each query eliminates half of the remaining possible address values, so the PDT has depth $k+1$, while the function's sparsity is exponential in $k$. To modify the function to prevent this approach, we encode the address using subspaces to obfuscate it while maintaining Fourier sparsity.

\begin{example}[Subspace addressing]\label{ex:subspaceAddressing}
    Let $A_1, \ldots, A_{2^k} \subset \F_2^{7k}$ be disjoint affine subspaces of dimension $2k$. Define $f : \F_2^{7k+2^k} \to \{-1,1\}$ by 
    \[
    f(x,y) = \begin{cases}
        (-1)^{y_i} & x \in A_i \\
        1 & x \not\in A_1 \cup \cdots \cup A_{2^k}
    \end{cases},
    \]
    where $x \in \F_2^{7k}$ and $y \in \F_2^{2^k}$.
\end{example}

We choose $A_i$'s randomly and show that the resulting function $f$ has the suitable properties we need with high probability. 

\begin{lemma}
\label{lemma:random-f}
    Suppose the random function $f$ is constructed by picking random affine subspaces $A_1,\cdots,A_{2^k}\subset \F_2^{7k}$ as follows: for each $i\in [2^k]$, choose vectors $a_i,v_i^1,\cdots,v_i^{2k}\in \F_2^{7k}$ uniformly and independently, and let $V_i = \langle v_i^1,\cdots,v_i^{2k}\rangle$ and $A_i =  V_i+a_i$. Then with probability $1-2^{-k+2}$, all of the following hold:
    \begin{enumerate}
        \item[(a)]  $\forall i$, $\dim(V_i) = 2k$.
        \item[(b)]  $\forall i\neq j$, $A_i\cap A_j = \emptyset$.
        \item[(c)]  $ \forall i\neq j$, $V_i\cap V_j = \{0\}$.
        \item[(d)] For all nonzero  
        $v\in \F_2^{7k}$, $|\{i : v \in V_i^{\perp}\}|\leq 7$.
    \end{enumerate}
\end{lemma}

\begin{proof} For brevity, let $m = 7k$.
\begin{enumerate}
    \item[(a)]  Fix $i \in [2^k]$. The probability that vectors $v_i^1,\cdots,v_i^{2k}$ are linearly independent is at least 
    $$\frac{2^m-1}{2^m}\cdot \frac{2^m-2}{2^m}\cdot \frac{2^m-2^2}{2^m}\cdot\cdots \cdot \frac{2^m-2^{2k-1}}{2^m}\geq (1-2^{2k-m})^m\geq 1- m2^{2k-m}.$$
    Hence the probability that there is $i\in [2^k]$ for which $v_i^1,\cdots,v_i^{2k}$  are not linearly independent is at most $m2^{3k-m} = 7k 2^{-4k} \le 2^{-k}$.

    \item [(b)] Fix $i\neq j$. The probability that $A_i\cap A_j \neq \emptyset$ is at most $2^{2k}2^{2k-m} = 2^{4k-m}$. 
    Hence, the probability that there are $i\neq j$ with $A_i\cap A_j\neq \emptyset$ is at most $2^{2k}2^{4k-m} = 2^{-k}$.
    
    \item [(c)] Fix $i\neq j$. The probability that $V_i\cap V_j \neq \{0\}$ is at most $(2^{2k}-1)2^{2k-m} \leq 2^{4k-m}$. 
    Hence, the probability that there are $i\neq j$ with $V_i\cap V_j\neq \emptyset$ is at most $2^{2k}2^{4k-m} = 2^{-k}$.
    
    \item [(d)] The probability  that a fixed nonzero vector $v \in \F_2^{7k}$ is orthogonal to at least $t$ subspaces among $V_1,\cdots,V_{2^k}$ is at most $\binom{2^k}{t}2^{-2tk} \le 2^{-tk}$. Taking $t=8$ and union bounding over all $2^{7k}-1$ options for $v$ shows that the probability that there is $v$ for which $|\{i : v \in V_i^{\perp}\}|\geq 8$ is at most 
    $2^{-k}$. 
    \end{enumerate}
    
    By the union bound, the probability that any of items (a) to (d) are not satisfied is at most $4\cdot 2^{-k} = 2^{-k+2}$. 
\end{proof}

We will assume from now on that $f$ is chosen randomly so that \Cref{lemma:random-f} holds, and set $\S=\supp(\widehat{f})$.
It remains to prove there is no large folding direction. First, we give a lower bound on the size of Fourier support of $f$.
\begin{claim}\label{claim:sizeOfS}
$|\mathcal{S}|\geq 2^{6k}$.
\end{claim}
\begin{proof}
We can express $f$ as
\begin{align*}
f(x,y) &= \mathbbm{1}_{(A_1 \cup \cdots \cup A_{2^k})^c}(x) + \sum_{i=1}^{2^k} \mathbbm{1}_{ A_i}(x) \cdot (-1)^{y_i} 
\\
 &= 1 - \sum_{i=1}^{2^k} \mathbbm{1}_{A_i}(x) + \sum_{i=1}^{2^k} \mathbbm{1}_{A_i}(x) \cdot (-1)^{y_i}\\
 &= 1 + \sum_{i=1}^{2^k} \mathbbm{1}_{ A_i}(x) \cdot ((-1)^{y_i}-1).
\end{align*}

By \Cref{fct:Fourier_subspace}, the Fourier support of the function $\mathbbm{1}_{ A_i}(x)$  is  $ V_i^\perp\subset \F_2^{7k}$, and of $\mathbbm{1}_{ A_i}(x)\cdot (-1)^{y_i}$ is $ V_i^\perp+e_i$,
where $e_i$ is the $i$-th basis vector in the standard basis for $\F_2^{2^k}$ embedded in the space $\F_2^{7k}\times \F_2^{2^k}$.
Since the affine subspaces $ V_i^\perp+e_i$ are disjoint and also $\left( V_i^\perp+e_i\right) \cap \left( V_i^\perp\right) = \emptyset $  the coefficients of characters in $V_i^\perp+e_i$ will not be canceled. Hence, we get that $$\bigcup_{i=1}^{2^k} \left(V_i^\perp + e_i\right)\subset \mathcal{S} $$ and so $|\mathcal{S}|\geq 2^k \cdot 2^{7k - 2k} = 2^{6k}$.
\end{proof}

We also need the following claim.
\begin{claim}\label{claim:intersectionsize}
    Suppose $W_1,W_2\subset \F_2^{n}$ are two linear subspaces such that $W_1\cap W_2 = \{0\}$. Then for all $x\in \F_2^n$, $$|W_1^\perp\cap (W_2^\perp +x)|= 2^{n-\dim{W_1}-\dim{W_2}}.$$
\end{claim}
\begin{proof}
Suppose $\dim(W_1) = d_1$ and $\dim(W_2) = d_2$. 
Without loss of generality, assume that $W_1 = \F_2^{d_1}\times 0^{d_2}\times 0^{n-d_1-d_2}$ and $W_2  = 0^{d_1}\times \F_2^{d_2}\times 0^{n-d_1-d_2}$. 
Note that 
$W_1^\perp = 0^{d_1}\times \F_2^{d_2}\times \F_2^{n-d_1-d_2}$ 
and 
$W_2^\perp = \F_2^{d_1}\times 0^{d_2}\times \F_2^{n-d_1-d_2}$. Pick an arbitrary $x=(x_1,x_2,x_3)\in \F_2^{d_1}\times \F_2^{d_2}\times \F_2^{n-d_1-d_2}$. Then $W_2^\perp + (x_1,x_2,x_3) = \F_2^{d_1}\times \{x_2\}\times \F_2^{n-d_1-d_2}$ and 
$W_1^\perp\cap (W_2^\perp +x) = 0^{d_1}\times \{x_2\}\times \F_2^{n-d_1-d_2} $ has the claimed size.
\end{proof}
Finally, \cref{thm:preciseGreedy} follows from  claim below.
\begin{claim}
    For all distinct $\gamma_1, \gamma_2 \in \F_2^{7k+2^k}$, we have
    $$|(\S+\gamma_1) \cap (\S + \gamma_2)| \le 2^{5k+4}.$$ 
\end{claim}

\begin{proof}
First, note that it suffices to prove the claim for all distinct $\gamma_1,\gamma_2 \in \S$, since if $s_1 + \gamma_1 = s_2 + \gamma_2$ for $s_1, s_2 \in \S$, it must be that $\gamma_1+\gamma_2 = s_1+s_2 \in \S+\S$. Pick an arbitrary non-zero $\gamma = \gamma_1 + \gamma_2$ for $\gamma_1, \gamma_2\in \mathcal{S}$.  
Remember that 
\[
|(\S+\gamma_1) \cap (\S + \gamma_2)| = |\S \cap (\S + \gamma)| \quad \text{and} \quad \mathcal{S} \subseteq \left(\bigcup_{i=1}^{2^k}V_i^\perp\right)\cup\left( \bigcup_{i=1}^{2^k}(V_i^\perp+e_i)\right).
\]
Hence $\S \cap (\S + \gamma)\subseteq A\cup B\cup C$, where 
\begin{align*}
    A &= \bigcup_{i,j} \left(V_i^\perp \cap (V_j^\perp+\gamma)\right) \\
    B &= \bigcup_{i,j} \left(V_i^\perp \cap( V_j^\perp+e_j+\gamma)\right) \\
    C &= \bigcup_{i,j} \left((V_i^\perp+e_i) \cap (V_j^\perp+e_j+\gamma)\right).
\end{align*}

Let $|\cdot|$ denote the Hamming weight of a vector.
Decompose $\gamma = (\gamma_x,\gamma_y)$  where $\gamma_x\in\F_2^{7k}$ and $\gamma_y\in\F_2^{2^k}$.
    Observe that $|\gamma_y|\leq 2$, since (as noted above) we may assume $\gamma \in \S + \S$ without loss of generality.
    \begin{itemize}[leftmargin=1.6cm]
    
    \item[\textbf{Case 1:}]
     $|\gamma_y| = 0$. 

Note that in this case $B = \emptyset$ and $C = \bigcup_{i} \left((V_i^\perp+e_i) \cap (V_i^\perp+e_i+\gamma_x)\right)$. Overall, we get 
   \begin{align*}
   |\S \cap (\S + \gamma_x)|&\leq  
\left|\bigcup_{i,j} \left(V_i^\perp \cap (V_j^\perp+\gamma_x)\right) \right|
+
\left|\bigcup_{i} \left((V_i^\perp+e_i) \cap (V_i^\perp+e_i+\gamma_x)\right) \right|\\
&=  
\left|\bigcup_{i,j} \left(V_i^\perp \cap (V_j^\perp+\gamma_x)\right) \right|
+
\left|\bigcup_{i} \left(V_i^\perp \cap (V_i^\perp+\gamma_x)\right) \right|\\
     &\leq  \sum_{i\neq j} |V_i^\perp \cap (V_j^\perp+\gamma_x)| + 2\sum_{i}|V_i^\perp \cap (V_i^\perp+\gamma_x)|\\
     &\leq \sum_{i\neq j} |V_i^\perp \cap (V_j^\perp+\gamma_x)|+2\cdot 2^{5k} \cdot |\{i : \gamma_x \in V_i^{\perp}\}|.
\end{align*}
To bound the first term, note that $V_i\cap V_j = \{0\}$ for all $i\neq j$ (by item (c) of \cref{lemma:random-f}). Using \cref{claim:intersectionsize} we have that $|V_i^\perp \cap (V_j^\perp+\gamma_x)| = 2^{7k-\dim(V_i)-\dim(V_j)}= 2^{7k-2k-2k}= 2^{3k}$.

    To bound the second term, by item (d) of \cref{lemma:random-f}, we have that  $|\{i : \gamma_x \in V_i^{\perp}\}|\leq 7$.
    Overall, we get that 
    $$|\S \cap (\S + \gamma)| \leq 2^{2k}\cdot2^{3k}+7\cdot 2^{5k+1} \le 2^{5k+4}.$$
    
    \item[\textbf{Case 2:}]
     $|\gamma_y| = 1$. Suppose  $\gamma_y = e_i$ for some $i$. 
    
    In this case, $A = C = \emptyset$ and $B = V_i^\perp\cap(V_i^\perp+e_i+\gamma_y)$. Hence,
    \[
    |\S \cap (\S + \gamma)| \le |V_i^\perp\cap(V_i^\perp+e_i+\gamma_y)| \leq |V_i^\perp|= 2^{5k},
    \]

    \item[\textbf{Case 3:}]
    $|\gamma_y| = 2$. This is similar to Case 2. \qedhere

    \end{itemize}
\end{proof}
\begin{remark}\label{rmk:improvedBounds}
    We have chosen parameters for simplicity of exhibition; however, by choosing the original disjoint affine subspaces from $\F_2^{(6+\varepsilon)k}$ rather than $\F_2^{7k}$, a similar analysis rules out any bounds stronger than $\PDT(f) = \widetilde{O}(|\S|^{1/5})$ resulting from this greedy method. 
\end{remark}

\section{Many good folding directions}\label{sec:many}
Rather than hoping for one large folding direction, \cite{mande2020parity} sought many nontrivial ones.
In this section, we refute their conjecture (\cref{conj:manydir}) with the following quantified version of \cref{thm:mainthm}.
\begin{theorem}\label{thm:quantifiedMain}
    For $n=2^d-1$ with $d \in \mathbb{N}$, there is a function $f:\F_2^n\to\{-1,1\}$ such that for $\S = \supp(\widehat{f})$, it holds
    \[
    \Pr_{\gamma_1, \gamma_2 \in \S} \left[\left|(\S + \gamma_1) \cap (\S + \gamma_2)\right| \ge 2^{k+2} \right] \le 2^{-k} + 2^{1-d} \qquad \forall k \ge 1.
    \]
\end{theorem}
In our construction, $|\S| = \poly(n)$, which is the primary regime of interest. For larger $\S$, say of size $|\S| = \exp(n^c)$ for some constant $c > 0$, the log-rank conjecture is trivially true, since $n < \polylog(|\S|)$.

Let $T$ be a full binary decision tree of depth $d$. There are $n = 2^d-1$ internal nodes indexed by $[2^d-1]$, where we query (distinct) $x_i$ at node $i$. Each of the largest depth internal nodes $v$ is adjacent to two leaves: -1 and 1, corresponding to $v = 0$ and $v=1$, respectively. Let $f :\mathbb{F}_2^{n}\to \{-1,1\}$ be the resulting function. For example, the following decision tree corresponds to $f$ for $n=7$.
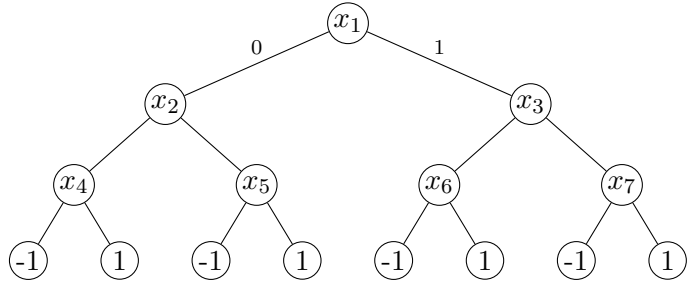
\begin{center}
\begin{forest}
for tree={
    grow=south,
    circle, draw, minimum size=3ex, inner sep=1pt,
    s sep=7mm
        }
[$x_1$
    [$x_2$, edge label={node[midway,above,font=\scriptsize]{0}}
        [$x_4$
            [-1]
            [1]
        ]
        [$x_5$
            [-1]
            [1]
        ]
    ]
    [$x_3$, edge label={node[midway,above,font=\scriptsize]{1   }}
        [$x_6$
            [-1]
            [1]
        ]
        [$x_7$
            [-1]
            [1]
        ]
    ]
]
\end{forest}
\end{center}

As we will soon show, the Fourier support of $f$ corresponds to (subsets of) paths down the tree, where $\left|(\S + \gamma_1) \cap (\S + \gamma_2)\right|$ is determined by the lowest common ancestor of the paths of $\gamma_1$ and $\gamma_2$. Since it is overwhelmingly likely the two paths will quickly diverge, we find $\left|(\S + \gamma_1) \cap (\S + \gamma_2)\right|$ is typically small. 

Suppose the leaves are indexed by $[2^d]$.
Then $f$ can be written as 
\begin{equation}\label{fsum}
    \sum_{i\in [2^d]}\sign(L_i)\cdot \mathbbm{1}_{L_i},
\end{equation}
where $\mathbbm{1}_{L_i}$ denotes the indicator function of the inputs that result in leaf $i$, and $\sign(L_i)\in \{-1,1\}$ is the output at leaf $i$. Let $P_i$ be the ordered set of coordinates that are queried to reach the leaf $i$. Then for input $x = (x_1,\ldots,x_n)\in \mathbb{F}_2^n$, we can write 
\[
    \mathbbm{1}_{L_i}(x) = \prod_{t\in P_i} \left(\frac{1+ (-1)^{a_t + x_t}}{2}\right) = \frac{1}{2^d}\left(\sum_{P\subseteq P_i}(-1)^{\sum_{j \in P}a_j} \cdot (-1)^{\sum_{j \in P}x_j}\right),
\]
where $a_t\in \mathbb{F}_2$ is the output of node $t$ on the path $P_i$.

To find the Fourier support $\S=\supp(\widehat{f})$, it remains to determine which terms ``survive'' cancellation in Equation \eqref{fsum}. Let $\N(i)$ be the index of the internal node adjacent to leaf $i$. Observe that when $\N(i) = \N(j)$ for $i \ne j$ (so $\sign(L_i) = -\sign(L_j)$),
\begin{align*}
    2^d (\sign(L_i)\cdot \mathbbm{1}_{L_i}(x) + \sign(L_j)\cdot \mathbbm{1}_{L_j}(x)) &= \sign(L_i) \sum_{P\subseteq P_i}(-1)^{\sum_{t \in P}a_t} \cdot (-1)^{\sum_{t \in P}x_t} \\ 
    &\qquad - \sign(L_i) \sum_{P\subseteq P_j}(-1)^{\sum_{t \in P}a_t} \cdot (-1)^{\sum_{t \in P}x_t} \\
    &= 2 \cdot \sign(L_i) \cdot \sum_{P \subseteq P_i \: : \: \N(i) \in P}(-1)^{\sum_{t \in P}a_t} \cdot (-1)^{\sum_{t \in P}x_t},
\end{align*}
 since $x_{\N(i)}$ is the only $x$ value that $P_i$ and $P_j$ disagree on.
That is, each term in $f$'s expansion must contain $\N(i)$ for some $i$. Moreover, once these cancellations are made, $\mathbbm{1}_{L_i}$ does not interact with $\mathbbm{1}_{L_j}$ for $\N(i) \ne \N(j)$, since no term can contain both $\N(i)$ and $\N(j)$. In summary,
\[
\S = \bigcup_{i \in [2^d]} \{s : s \subseteq P_i \text{ and } \N(i) \in s\}.
\]

Let $\gamma_1, \gamma_2 \in \S$. By our observation on the structure of $\S$, they have the form $\gamma_1 = \alpha_1 \dot\cup \{\N(i)\}$ and $\gamma_2 = \alpha_2 \dot\cup \{\N(j)\}$ for some $i,j \in [2^d]$.
We are interested in the number of pairs $(\beta_1, \beta_2) \in \S \times \S$ such that $\gamma_1 + \gamma_2 = \beta_1 + \beta_2$. It will suffice to focus on the setting $\N(i) \ne \N(j)$ since this occurs with overwhelming probability. In this case, the quantity $|(\S + \gamma_1) \cap (\S + \gamma_2)|$ depends only on the depth of the lowest common ancestor of $P_i$ and $P_j$. 
\begin{claim}
    If $|(\S + \gamma_1) \cap (\S + \gamma_2)| \ge 2^{k+2}$, then the lowest common ancestor of $P_i$ and $P_j$ is at depth at least $k$.
\end{claim}
\begin{proof}
    We will show the contrapositive. Suppose the lowest common ancestor $a$ of $P_i$ and $P_j$ is at depth $\ell < k$, and suppose $\beta_1, \beta_2 \in \S$ satisfy $\beta_1 + \gamma_1 = \beta_2 + \gamma_2$. Without loss of generality, assume $\N(i) \in \beta_1$ and $\N(j) \in \beta_2$. Then $\beta_1$ and $\beta_2$ must be a subset of the elements in the paths $P_i$ and $P_j$, respectively. 
    
    First, consider each element $E \in P_i \cap P_j$, which is all those above (and including) $a$. If $E \in \gamma_1 + \gamma_2$, then $E \in \beta_1 + \beta_2$ only if $E$ is in precisely one of $\beta_1, \beta_2$. Likewise, if $E \not\in \gamma_1 + \gamma_2$, then $E \not\in \beta_1 + \beta_2$ only if $E$ is in neither or both of $\beta_1, \beta_2$. In either case, we have two options for each $E$. 

    Now consider each element $E \in P_i$ below $a$. By assumption, $E \not\in P_j$. Thus, if $E \in \gamma_1 + \gamma_2$, it must be that $E \in \gamma_1$ and $E \not\in \gamma_2$. For $\beta_1 + \beta_2$ to contain $E$, we must likewise have $E \in \beta_1$ and $E \not\in \beta_2$. Similarly, if $E \not\in \gamma_1 + \gamma_2$, it cannot be in $\gamma_1$ or $\gamma_2$. Thus, it is not in $\beta_1$ or $\beta_2$ either. An identical argument for $E \in P_j$ shows that we only have one way to account for elements in the paths $P_i$ or $P_j$ below $a$.

    Doubling to compensate for the cases where $\N(j) \in \beta_1$ and $\N(i) \in \beta_2$, we find the number of options for $(\beta_1, \beta_2) \in \S \times \S$ such that $\beta_1+\gamma_1 = \beta_2+\gamma_2$ is at most $2^{\ell+2}<2^{k+2}$.
\end{proof}

\cref{thm:quantifiedMain} follows quickly from the claim. The probability that $P_i$ and $P_j$ have a common ancestor at depth at least $k$ is at most $2^{-k}$, so 
\begin{align*}
    &\quad\,\Pr_{\gamma_1, \gamma_2 \in \S} \left[\left|(\S + \gamma_1) \cap (\S + \gamma_2)\right| \ge 2^{k+2} \right] \\ 
    &\le \Pr_{\gamma_1, \gamma_2 \in \S} \left[\left|(\S + \gamma_1) \cap (\S + \gamma_2)\right| \ge 2^{k+2} \:\middle\vert \:\N(\gamma_1) \ne \N(\gamma_2)\right] + 2^{1-d} \\
    &\le 2^{-k} + 2^{1-d},
\end{align*}
where we overload notation by letting $\N(\gamma) = \N(i) \in \gamma$.

\section{Conclusion}\label{sec:conc}
While the provided functions rule out specific approaches, it is worth noting that neither are a counterexample to the log-rank conjecture. The subspace addressing function (\Cref{sec:one}) has a simple PDT: first individually query all $7k$ address bits, then query the bit to the corresponding subspace. Since the Fourier sparsity is at least $2^{6k}$, this is certainly affordable.
While this example refutes a general greedy approach, such an approach works for the decision tree function (\Cref{sec:many}). Each query of the root variable eliminates half the paths (and thus reduces the sparsity by two), so iterating this process quickly makes the function constant.

\section{Acknowledgements}
This work was done in part while the authors were visiting the Simons Institute for the Theory of Computing. A.O. thanks Daniel M. Kane for a number of helpful discussions. We also thank anonymous reviewers for useful comments on earlier versions of this manuscript.

\bibliographystyle{amsalpha}  
\bibliography{biblio} 

\newcommand{\etalchar}[1]{$^{#1}$}
\providecommand{\bysame}{\leavevmode\hbox to3em{\hrulefill}\thinspace}
\providecommand{\MR}{\relax\ifhmode\unskip\space\fi MR }
\providecommand{\MRhref}[2]{%
  \href{http://www.ams.org/mathscinet-getitem?mr=#1}{#2}
}
\providecommand{\href}[2]{#2}
\begin{thebibliography}{OWZ{\etalchar{+}}14}

\bibitem[CGS21]{chattopadhyay2020towards}
Arkadev Chattopadhyay, Ankit Garg, and Suhail Sherif, \emph{Towards stronger
  counterexamples to the log-approximate-rank conjecture}, Proceedings of the
  41st IARCS Annual Conference on Foundations of Software Technology and
  Theoretical Computer Science (FSTTCS), 2021, pp.~175--190.

\bibitem[HHL18]{hatami2018structure}
Hamed Hatami, Kaave Hosseini, and Shachar Lovett, \emph{Structure of protocols
  for {XOR} functions}, SIAM Journal on Computing \textbf{47} (2018), no.~1,
  208--217.

\bibitem[KLMY21]{knop2021log}
Alexander Knop, Shachar Lovett, Sam McGuire, and Weiqiang Yuan, \emph{Log-rank
  and lifting for {AND}-functions}, Proceedings of the 53rd Annual ACM
  Symposium on Theory of Computing (STOC), 2021, pp.~197--208.

\bibitem[Lov16]{lovett2016communication}
Shachar Lovett, \emph{Communication is bounded by root of rank}, Journal of the
  ACM (JACM) \textbf{63} (2016), no.~1, 1--9.

\bibitem[LS93]{lovasz1993communication}
L{\'a}szl{\'o} Lov{\'{a}}sz and Michael Saks, \emph{Communication complexity
  and combinatorial lattice theory}, Journal of Computer and System Sciences
  \textbf{47} (1993), no.~2, 322--349.

\bibitem[MO09]{montanaro2009communication}
Ashley Montanaro and Tobias Osborne, \emph{On the communication complexity of
  {XOR} functions}, arXiv preprint arXiv:0909.3392 (2009).

\bibitem[MS82]{mehlhorn1982vegas}
Kurt Mehlhorn and Erik~M Schmidt, \emph{Las {V}egas is better than determinism
  in {VLSI} and distributed computing}, Proceedings of the Fourteenth Annual
  ACM Symposium on Theory of Computing (STOC), 1982, pp.~330--337.

\bibitem[MS20]{mande2020parity}
Nikhil~S Mande and Swagato Sanyal, \emph{On parity decision trees for
  {F}ourier-sparse boolean functions}, 40th IARCS Annual Conference on
  Foundations of Software Technology and Theoretical Computer Science (FSTTCS),
  Schloss Dagstuhl-Leibniz-Zentrum f{\"u}r Informatik, 2020.

\bibitem[O'D14]{o2014analysis}
Ryan O'Donnell, \emph{Analysis of boolean functions}, Cambridge University
  Press, 2014.

\bibitem[OWZ{\etalchar{+}}14]{odonnell2014composition}
Ryan O'Donnell, John Wright, Yu~Zhao, Xiaorui Sun, and Li-Yang Tan, \emph{A
  composition theorem for parity kill number}, 2014 IEEE 29th Conference on
  Computational Complexity (CCC), IEEE, 2014, pp.~144--154.

\bibitem[RY20]{rao2020communication}
Anup Rao and Amir Yehudayoff, \emph{Communication complexity: and
  applications}, Cambridge University Press, 2020.

\bibitem[ST23]{sudakov2023matrix}
Benny Sudakov and István Tomon, \emph{Matrix discrepancy and the log-rank
  conjecture}, 2023.

\bibitem[STV14]{shpilka2014structure}
Amir Shpilka, Avishay Tal, and Ben~Lee Volk, \emph{On the structure of boolean
  functions with small spectral norm}, Proceedings of the 5th Conference on
  Innovations in Theoretical Computer Science (ITCS), 2014, pp.~37--48.

\bibitem[TWXZ13]{tsang2013fourier}
Hing~Yin Tsang, Chung~Hoi Wong, Ning Xie, and Shengyu Zhang, \emph{Fourier
  sparsity, spectral norm, and the log-rank conjecture}, 2013 IEEE 54th Annual
  Symposium on Foundations of Computer Science (FOCS), IEEE, 2013,
  pp.~658--667.

\bibitem[ZS10]{zhang2010parity}
Zhiqiang Zhang and Yaoyun Shi, \emph{On the parity complexity measures of
  boolean functions}, Theoretical Computer Science \textbf{411} (2010),
  no.~26-28, 2612--2618.

\end{thebibliography}

\end{document}